\documentclass{llncs}
\usepackage{latexsym,graphics,amssymb,setspace,amsmath}
\usepackage[pdftex]{graphicx}
\newcommand{\Id}[1]{\ensuremath{\mathit{#1}}}

\newcommand{\set}[1]{\left\{ #1\right\}}

\newcommand{\realrange}[2]{\left[#1, #2\right]}

\newcommand{\unitrange}[2]{\realrange{0}{1}}

\newcommand{\Oh}[1]{\mathcal{O}\!\left( #1\right)}

\newcommand{\llabel}[1]{\label{\labelprefix:#1}}
\newcommand{\labelprefix}{} 

\newcommand{\labelcommand}{}
\newcommand{\captiontext}{}
\newsavebox{\buchalgorithmparam}
\newcounter{lineNumber}
\newenvironment{buchalgorithmpos}[3]{%
\renewcommand{\labelcommand}{#2}%
\renewcommand{\captiontext}{#3}%
\sbox{\buchalgorithmparam}{\parbox{\textwidth}{#3}}%
\begin{figure}[#1]\begin{center}\begin{code}\setcounter{lineNumber}{1}}{%
\end{code}\end{center}\caption{\llabel{\labelcommand}\captiontext}\end{figure}}

{\end{buchalgorithmpos}}

\newenvironment{code}{\noindent\it%
\begin{tabbing}%
\hspace{1.5em}\=\hspace{1.5em}\=\hspace{1.5em}\=\hspace{1.5em}\=\hspace{1.5em}\=%
\hspace{1.5em}\=\hspace{1.5em}\=\hspace{1.5em}\=\hspace{1.5em}\=\hspace{1.5em}\=%
\kill}{\end{tabbing}}

\newcommand{\Is}{\mbox{\rm := }}

\newdimen\endofsize\endofsize=0.5em
\def\endofbeweis{~\quad\hglue\hsize minus\hsize
                 \hbox{\vrule height \endofsize width
\endofsize}\par}

\newcommand{\donotshow}[1]{}

\newcommand{\ignore}[1]{}

\newcommand{\mbegin}{\{\ \ }

\newlength{\mleftindent}
\newlength{\mindent}
\newlength{\mboxwidth}

\newlength{\preprogramskip}
\newlength{\postprogramskip}
\newlength{\mexpwidth}
\newlength{\mexpindent}

\newlength{\proofpostskipamount}
\newlength{\proofpreskipamount}
               {\par\vspace{\proofpreskipamount}\noindent{\bf Proof:}\hspace{0.5em}}
               {\nopagebreak%
                \strut\nopagebreak%
                \hspace{\fill}\qed\par\vspace{\proofpostskipamount}\noindent}

               {\par\vspace{0.5ex}\noindent{\bf Proof #1:}\hspace{0.5em}}%
               {\nopagebreak%
                \strut\nopagebreak%
                \hspace{\fill}\qed\par\medskip\noindent}

\newcommand{\myurl}[1]{{\footnotesize \url{#1}}}

\usepackage{tikz}
\usepackage{subfigure}

\begin{document}

\title{Efficient Parallel and External Matching}

\author{Marcel Birn, Vitaly Osipov, Peter Sanders, Christian Schulz, Nodari Sitchinava}
\institute{Karlsruhe Institute of Technology, Karlsruhe, Germany \\
\email{marcelbirn@gmx.de,\{{osipov,sanders,christian.schulz\}@kit.edu,nodari@ira.uka.de}}}


\maketitle

\begin{abstract}
We show that a simple algorithm for computing a matching on a graph runs in a
logarithmic number of phases incurring work linear in the input size. The
algorithm can be adapted to provide efficient algorithms in several models of
computation, such as PRAM, External Memory, MapReduce and distributed memory
models. Our CREW PRAM algorithm is the first $\Oh{\log^2 n}$ time,
linear work algorithm. Our experimental results indicate the algorithm's high speed and
efficiency combined with good solution quality.  
\end{abstract}

\section{Introduction}

A matching $M$ of a graph $G=(V,E)$ is a subset of edges such that no two
elements of $M$ have a common end point. Many applications require the
computation of matchings with certain properties, like being maximal (no edge can
be added to $M$ without violating the matching property), having maximum
cardinality, or having maximum total weight $\sum_{e\in M}w(e)$. Although these
problems can be solved optimally in polynomial time, optimal algorithms are not
fast enough for many applications involving large graphs where we need near
linear time algorithms. For example, the most efficient algorithms for graph
partitioning rely on repeatedly contracting maximal matchings, often trying to
maximize some edge rating function $w$. Refer to \cite{HSS10} for details
and examples. For very large graphs, even linear time is not enough -- we
need a parallel algorithm with near linear work or an algorithm working in the
external memory model \cite{VitShr94}.

Here we consider the following simple \emph{local max} algorithm
\cite{Hoepman04}: Call an edge locally maximal, if its weight is larger than the
weight of any of its incident edges; for unweighted problems, assign unit
weights to the edges. When comparing edges of equal weight, use tie breaking
based on random perturbations of the edge weights.  The algorithm starts with an
empty matching $M$. It repeatedly adds locally maximal edges to $M$ and removes
their incident edges until no edges are left in the graph.  The result is
obviously a maximal matching (every edge is either in $M$ or it has been removed
because it is incident to a matched edge). The algorithm falls into a family of
weighted matching algorithms for which Preis \cite{Preis99} shows that they
compute a $1/2$-approximation of the maximum weight matching problem.  Hoepman
\cite{Hoepman04} derives the local max algorithm as a distributed adaptation of
Preis' idea. Based on this, Manne and Bisseling \cite{ManneB07} devise
sequential and parallel implementations. They prove that the algorithm needs
only a logarithmic number of iterations to compute maximal matchings by noticing
that a maximal matching problem can be translated into a maximal independent set
problem on the \emph{line graph} which can be solved by Luby's algorithm
\cite{Luby86}.  However, this does not yield an algorithm with linear work since
it is not proven that the edge set indeed shrinks geometrically.%
\footnote{Manne and Bisseling show such a shrinking property under an assumption
  that unfortunately does not hold for all graphs.}  Manne and Bisseling also
give a sequential algorithm running in time $\Oh{m\log\Delta}$ where $\Delta$ is
the maximum degree. On a NUMA shared memory machine with 32 processors (SGI Origin 3800) they get
relative speedup $<6$ for a complete graph and relative speedup $\approx 10$ for a more sparse
graph partitioned with Metis. Since this graph still has average degree $\approx 200$
and since the speedups are not impressive this is a somewhat inconclusive result when one is interested in partitioning large sparse graphs on a larger number of processors.

Parallel matching algorithms have been widely studied. There is even a book on
the subject \cite{KarRyt98} but most theoretical results concentrate on
work-inefficient algorithms.
The only linear work parallel algorithm that we are aware of is a randomized
CRCW PRAM algorithm by Israeli and Itai \cite{IsraeliItai86}  which runs in
$\Oh{\log n}$ time and incurs linear work. Their algorithm, which we call IIM,
provably removes a constant fraction of edges in each iteration.

Fagginger Auer and Bisseling~\cite{FagBis12} study an algorithm similar to
\cite{IsraeliItai86} which we call red-blue matching (RBM) here. They implement
RBM on shared memory machines and GPUs.  They prove good shrinking behavior for
random graphs, however, provide no analysis for arbitrary graphs.

{\bf Our contributions.} We give a simple approach to implementing the local
max algorithm that is easy to adapt to many models of computation. We show that
for computing maximal matchings, the algorithm needs only linear work on a
sequential machine and in several models of parallel computation
(Section~\ref{s:parallel}).  Moreover it has low I/O complexity on several
models of memory hierarchies. 

Our CRCW PRAM local max algorithm matches the optimal asymptotic bounds of IIM.
However, our algorithm is simpler (resulting in better constant factors),
removes higher fraction of edges in each iteration (IIM's proof shows less than
5\% per iteration, while we show at least 50\%) and our analysis is a lot
simpler.  
We also provide the first CREW PRAM algorithm which runs in $\Oh{\log^2 n}$
time and linear work.\footnote{While a generic simulation of IIM on the CREW
PRAM model will result in a $\Oh{\log^2 n}$ time algorithm, the simulation
incurs $\Oh{n \log n}$ work due to sorting.}

In Section~\ref{s:experiments} we explain how to implement local max on
practical massively parallel machines such as MPI clusters and GPUs.  Our
experiments indicate that the algorithm yields surprisingly good quality for
the weighted matching problem and runs very efficiently on sequential machines,
clusters with reasonably partitioned input graphs, and on GPUs.  Compared to
RBM, the local max implementations remove more edges in each iteration  and
provide better quality results for the weighted case.  Some of the results
presented here are from the diploma thesis of Marcel Birn \cite{Birn12}.

\section{Parallel Local Max}\label{s:parallel}

Our central observation is:
\begin{lemma}\label{lem:remove}
  Each iteration of the local max algorithm for the unit weight case removes at
  least half of the edges in expectation.
\end{lemma}
\begin{proof}
  Consider the graph remaining in the currently considered iteration where
  $d(v)$ denotes the degree of a node and $m$ the remaining number of edges.
  Consider the end point at node $v$ of an edge $\set{u,v}$ as \emph{marked} if
  and only if some edge incident to $v$ becomes matched. Note that an edge is
  removed if and only if at least one of its end points becomes marked.  Now
  consider a particular edge $e=\set{u,v}$.  Since any of the $d(u)+d(v)-1$
  edges incident to $u$ and $v$ is equally likely to be locally maximal, $e$
  becomes matched with probability $1/(d(u)+d(v)-1)$.%
  \footnote{For this to be true, the random noise added for tie breaking needs
    to be renewed in every iteration. However, in our experiments this had no
    noticeable effect.}  If $e$ is matched, this event is responsible for setting
  $d(u)+d(v)$ marks, i.e., the expected number of marks caused by an edge is
  $(d(u)+d(v))/(d(u)+d(v)-1)\geq 1$. By linearity of expectation, the total
  expected number of marks is at least $m$.  Since no edge can have more than
  two marks, at least $m/2$ edges have at least one mark and are thus deleted.%
  \footnote{This is a conservative estimate. Indeed, if we make the
    (over)simplified assumption that $m$ marks are assigned randomly and
    independently to $2m$ end points, then only one fourth of the edges survives
    in expectation. Interestingly, this is the amount of reduction we observe in
    practice -- even for the weighted case.}\endofbeweis
\end{proof}

Assume now that each iteration can be implemented to run with work linear in the
number of surviving edges (independent of the number of nodes). Working naively
with the expectations, this gives us a logarithmic number of rounds and a
geometric sum leading to linear total work for computing a maximal
matching. This can be made rigorous by distinguishing \emph{good} rounds with at
least $m/4$ matched edges and bad rounds with less matched edges. By Markov's
inequality, we have a good round with constant probability. This is already
sufficient to show expected linear work and a logarithmic number of expected
rounds. We skip the details since this is a standard proof technique and since
the resulting constant factors are unrealistically conservative. An analogous
calculation for median selection can be found in
\cite[Theorem~5.8]{MehSan08}. One could attempt to show a shrinking factor close
to 1/2 rigorously by showing that large deviations (in the wrong direction) from
the expectation are unlikely (e.g., using Martingale tail bounds).  However this
would still be a factor two away from the more heuristic argument in Footnote~4
and thus we stick to the simple argument.

There are many ways to implement an iteration which of course depend on the
considered model of computation. 

{\bf Sequential Model.} 
For each node $v$ maintain a candidate edge $C[v]$, originally initialized to a
dummy edge with zero weight. In an iteration go through all remaining edges
$e=\set{u,v}$ three times. In the first pass, if $w(e)>w(C[u])$ set $C[u]\Is e$
(add random perturbation to $w(e)$ in case of a tie). If $w(e)>w(C[v])$ set
$C[v]\Is e$.  In the second pass, if $C[u]=C[v]=e$ put $e$ into the matching
$M$.  In the third pass, if $u$ or $v$ is matched, remove $e$ from the graph.
Otherwise, reset the candidate edge of $u$ and $v$ to the dummy edge.  Note
that except for the initialization of $C$ which happens only once before the
first iteration, this algorithm has no component depending on the number of
nodes and thus leads to linear running time in total if Lemma~\ref{lem:remove}
is applied.

{\bf CRCW PRAM Model.} In the most powerful variant of the \emph{Combining CRCW PRAM} that allows
concurrent writes with a maximum reduction for resolving write conflicts, the
sequential algorithm can be parallelized directly running in constant time per
iteration using $m$ processors. 

{\bf MapReduce Model.} The CRCW PRAM result together with the simulation result
of Goodrich et al.~\cite{goodrich:mapreduce} immediately implies that each
iteration of local max can be implemented in $\Oh{\log_M n}$ rounds and $\Oh{m
\log_M n}$ communication complexity in the MapReduce model, where $M$ is the
size of memory of each compute node. Since typical compute nodes in MapReduce have at
least $\Omega(m^{\epsilon})$ memory~\cite{karloff:mapreduce}, for some constant $\epsilon> 0$, each
iteration of local max can be performed in MapReduce in constant rounds and linear
communication complexity.

{\bf External Memory Models.} Using the PRAM emulation techniques for
algorithms with geometrically decreasing input size from
\cite[Theorem~3.2]{CGG+95} the above algorithm can be implemented in the
external memory~\cite{AggVit88} and cache-oblivious~\cite{FLPR99}
models in $\Oh{sort(m)}$ I/O complexity, which seems to be optimal.
\subsection{$\Oh{\log^2 n}$ work-optimal CREW solution}

In this section, we present a $\Oh{\log^2 n}$ CREW PRAM algorithm, which incurs only $\Oh{n+m}$ work.

\newcommand{\e}{{\tt E}}
\newcommand{\V}{{\tt V}}
\newcommand{\A} {{\tt A}}

Converting one representation of a graph into another can require as many as
$\Omega(m \log m)$ operations (e.g. the conversion from an unordered list of
edges into adjacency list representation requires sorting the edges). To
perform matching in $\Oh{n+m}$ work, we define a graph
representation suitable for our algorithm.

Consider an array $\V$ of $n$ elements, where each entry $\V[i] =
\sum_{j<i} deg(v_j)$ ($deg(v_j)$ denotes the degree of
node $v_j$).  An adjacency array representation of a graph is an array $\A$,
where entries $\A[\V[i]]$ through $\A[V[i+1]-1]$ are associated with vertex
$v_i \in G$ and store the edges incident on $v_i$.

We consider the following slightly altered adjacency array representation: the
edges are stored in a separate array $\e$ and the entries $\A[\V[i]]$ through
$\A[\V[i+1]-1]$ store the {\em pointers} to the corresponding edges in $\e$
(see Figure~\ref{fig:graph}).  Thus, we can view each entry of $\A$ as a tuple
$(v,e_k)$, where $v$ is a node in $G$ and $e_k$ is a pointer to a record
$\e[k]$ with information about the edge incident on $v$, such as the two
vertices of the edge, edge weight, or any other auxiliary information.

\begin{figure}[t]
\begin{center}
\includegraphics[scale=.8]{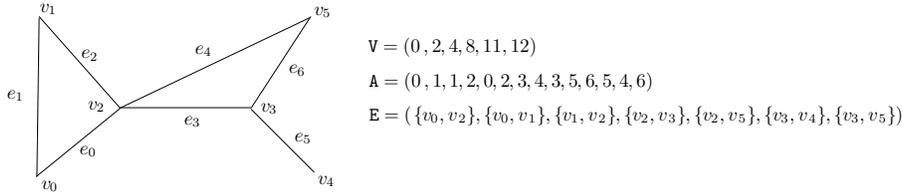}
\end{center}
\caption{Adjacency representation of a graph: Array $\e$ is a collection of
edges. Entries of $\A[\V[i]]$ through $\A[\V[i+1]-1]$ point to edges in $\e$
incident on vertex $v_i$. }
\label{fig:graph}
\end{figure}

Note that  any edge $\e[k] = \{v_i,v_j\}$ contains two corresponding entries in $\A$
pointing to it: $(v_i,e_k)$ and $(v_j,e_k)$.  During our algorithm, a processor
responsible for $(v_i,e_k)$ might need to find and update entry $(v_j,e_k)$ (and vice
versa).  The following lemma describes how to compute for each entry $(v_i,e_k)$
the address of the corresponding entry $(v_j,e_k)$ in $\A$.

\begin{lemma}
\label{lemma:reversal}
For every edge $e_k = \{v_i,v_j\} \in \e$ entries $(v_i, e_k)$ and $(v_j, e_k)$ of $\A$
can compute each other's index in $\A$ in $\Oh{1}$ time and $\Oh{|\A|}$ work in the
CREW PRAM model.
\end{lemma}

\begin{proof}
For every $\e[k] = \{v_i, v_j\}$ we show how $(v_j,e_k)$ can compute the address
of the corresponding entry $(v_i,e_k)$ in $\V$ for  $i < j$.  The addresses for
the other half of the entries are computed symmetrically.

The algorithm proceeds in two phases.  In the first phase, each entry
$(v_i,e_k)$, writes the address of $(v_i,e_k)$ in $\e[k] = \{v_i, v_j\}$, iff
$i < j$.  In the second phase, each entry $(v_j,e_k)$ reads the
address of $(v_i,e_k)$ from $\e[k] = \{v_i, v_j\}$ iff $j > i$.


If we assign a separate processor to each entry of $\A$, each processor
performs only $\Oh{1}$ steps.  Moreover, there are no concurrent writes because,
at each step only one of the two vertices of the edge $e_k$ writes to $\e[k]$.
Note, we need a concurrent read to $\e[k] = \{v_i, v_j\}$ to determine the
relative order of $i$ and $j$ for $v_i$ and $v_j$.
\endofbeweis \end{proof}

\begin{lemma}
\label{lemma:broadcast}
Using our graph representation, each node $v$ in the graph can apply an
associative operator $\oplus$ to all edges incident on $v$ in $\Oh{\log |\A|}$ time and
$\Oh{|\A|}$ work on the CREW PRAM model.
%
\end{lemma}

\begin{proof}
First, we read for each entry $(v, e_k) \in \A$ the value from $\e[k]$ on which
to apply the operator. Next,  we run segmented prefix sums with
$\oplus$ operator on these values, where segments are the portions of $\A$
representing the neighbors of a single node.  Finally, each entry of $(v, e_k)
\in \A$ applies its result of segmented prefix sums to the edge $\e[k]$, while using
the technique of Lemma~\ref{lemma:reversal} to avoid write conflicts. Each step
of the algorithm can be implemented in $\Oh{\log |\A|}$ time using $\Oh{|\A|}$
work.
%
\endofbeweis \end{proof}

Now we are ready to describe the solution to the matching problem.
We perform the following in each phase of the local max algorithm.

\begin{enumerate}
\item Each edge $e_k \in \e$ picks a random weight $w_k$.

\item Using Lemma~\ref{lemma:broadcast}, each vertex $v$ identifies the
heaviest edge $e_k$ incident on $v$ by applying   the associative operator {\sc max} to
the edge weights picked in the previous step.
\item Using Lemma~\ref{lemma:reversal}, each entry $(v_i, e_k)$
checks if $\e[k] = \{v_i,v_j\}$ is also the heaviest incident edge on $v_j$. If
so and $i < j$, $v_i$ adds $e_k$ to the matching and sets the deletion flag $f
= 1$ on $\e[k]$.

%

\item Using Lemma~\ref{lemma:broadcast}, each entry $(v_i, e_k)$ spreads the
deletion flag over all edges incident on $v_i$ by applying {\sc max} associative
operator on the deletion flags of incident edges on $v_i$. Thus, if at least
one edge incident on $v_i$ was added to the matching, all edges incident on
$v_i$ will be marked for deletion.
%

\item Now we must prepare the graph representation for the next phase by
removing all entries of $\e$ and $\A$ marked for deletion, compacting $\e$ and
$\A$ and updating the pointers of $\A$ to point to the compacted entries of
$\e$.  To perform the compaction, we compute for each entry $\e[k]$, how
many entries $\e[i]$ and $\A[i], i \le k$ must be deleted. This can be accomplished
using parallel prefix sums on the deletion flags of each entry in $\e$ and
$\A$. Let the result of prefix sums for edge $\e[k]$ be $d_k$ and for entry
$\A[i]$ be $r_i$. Then $k - d_k$ is the new address of the entry $\e[k]$ and $i
- r_i$ is the new address of $\A[i]$ once all edges marked for deletion are
  removed.

\item Each entry $\e[k]$ that is not marked for deletion copies itself to
$\e[k-d_k]$. The corresponding entry $(v, e_k) \in \A$ updates itself to
point to the new entry $\e[k-d_k]$, i.e., $(v, e_k)$ becomes $(v, e_{k-d_k})$,
and copies itself to $\A[i-r_i]$.

\end{enumerate}

The algorithm defines a single phase of the local max algorithm. Each step of
the phase takes at most $\Oh{\log (m+n)} = \Oh{\log n}$ time and
$\Oh{n+m}$ work in the CREW PRAM model.  Over $\Oh{\log m}$
phases, each with geometrically decreasing number of edges, the local max takes
$\Oh{\log^2 n}$ time and $\Oh{n+m}$ work in the CREW PRAM model.

\section{Implementations and Experiments}\label{s:experiments}

We now report experiments focusing on computing approximate maximum weight
matchings.
We consider the following families of inputs, where the first two classes allow comparison with the experiments from \cite{MauSan07}.

{\it Delaunay Instances}
are created by randomly choosing $n=2^x$ points in the unit square and computing
their Delaunay triangulation.
Edge weights are Euclidean distances.

{\it Random graphs} with $n:=2^x$ nodes, $\alpha n$ edges for $\alpha=\{4,16,64\}$, and random
edge weight chosen uniformly from $[0,1]$.

{\it Random geometric} graphs with $2^x$ nodes (rgg$x$). Each vertex is a random point in the unit square and edges connect vertices whose Euclidean distance is below 0.55 $\sqrt{\ln n/n}$. This threshold was chosen in order to ensure that the graph is almost connected.

{\it Florida Sparse Matrix.} Following \cite{FagBis12} we use $126$ symmetric non-0/1
matrices from \cite{UFsparsematrixcollection} using absolute values of their entries as edge
weights, see Appendix for the full list. The number of edges of the resulting graphs $m\in (0.5\ldots 16)\times 10^6$. See Appendix~\ref{app:instances} for a detailed list.

{\it Graph Contraction.} We use the graphs considered by KaFFPa for partitioning graphs from the 10'th DIMACS Implementation Challenge \cite{SanSch13a}. \\

We compare implementations of local max, the red-blue algorithm from
\cite{FagBis12} (RBM) (their implementation), heavy edge matching (HEM)
\cite{KarKum98b}, greedy, and the global path algorithm (GPA) \cite{MauSan07}.
HEM iterates through the nodes (optionally in random order) and matches the heaviest incident edge that is nonadjacent to a previously matched edge.
The greedy algorithm sorts the edges by decreasing weights,
scans them and inserts edges connecting unmatched nodes into the matching. GPA
refines greedy. It greedily inserts edges into a graph $G_2$ with
maximum degree two and no odd cycles. Using dynamic programming on the resulting
paths and even cycles, a maximum weight matching of $G_2$ is computed.

Sequential and shared-memory parallel experiments were performed on an Intel i7 920 $2.67$ GHz quad-core machine with $6$ GB of memory. We used a commodity NVidia Fermi GTX 480 featuring $15$ multiprocessors, each containing $32$ scalar processors, for a total of $480$ CUDA cores on chip. The GPU RAM is $1.5$~GB.
We compiled all implementations using CUDA $4.2$ and Microsoft Visual Studio $2010$ on $64$-bit Windows $7$ Enterprise with maximum optimization level.

\subsection{Sequential Speed and Quality}\label{s:sequential}

\begin{figure}[b]
\begin{center}
\includegraphics[scale=0.54]{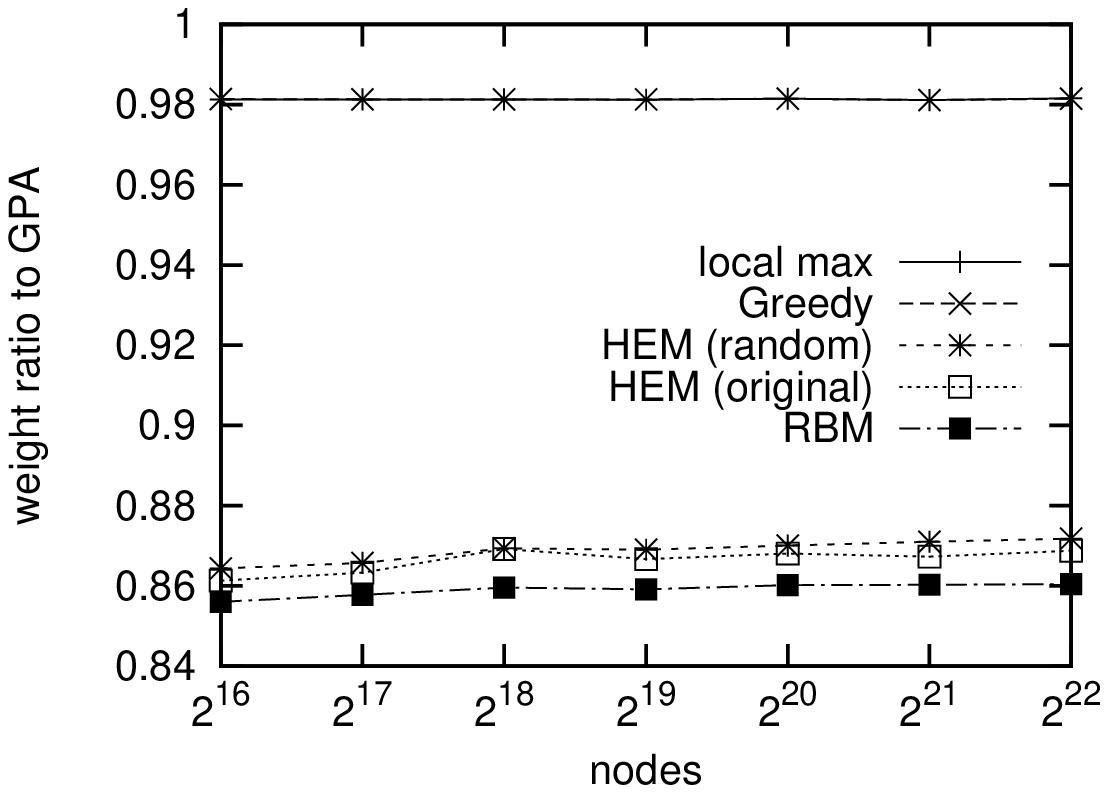}\includegraphics[scale=0.54]{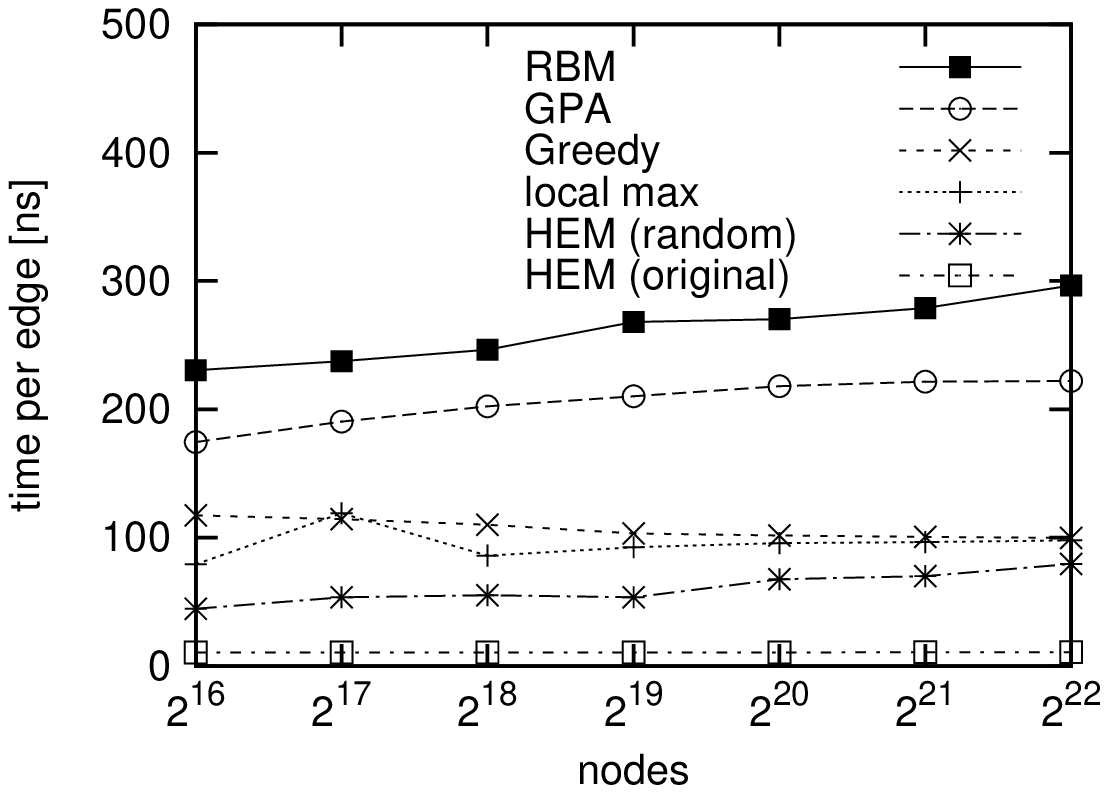}

\caption{\label{fig:delaunayquality}Ratio of the weights computed by GPA and other algorithms for Delaunay instances and running times.}
\end{center}
\end{figure}
\begin{figure}[tb]
\begin{center}
\includegraphics[scale=0.7]{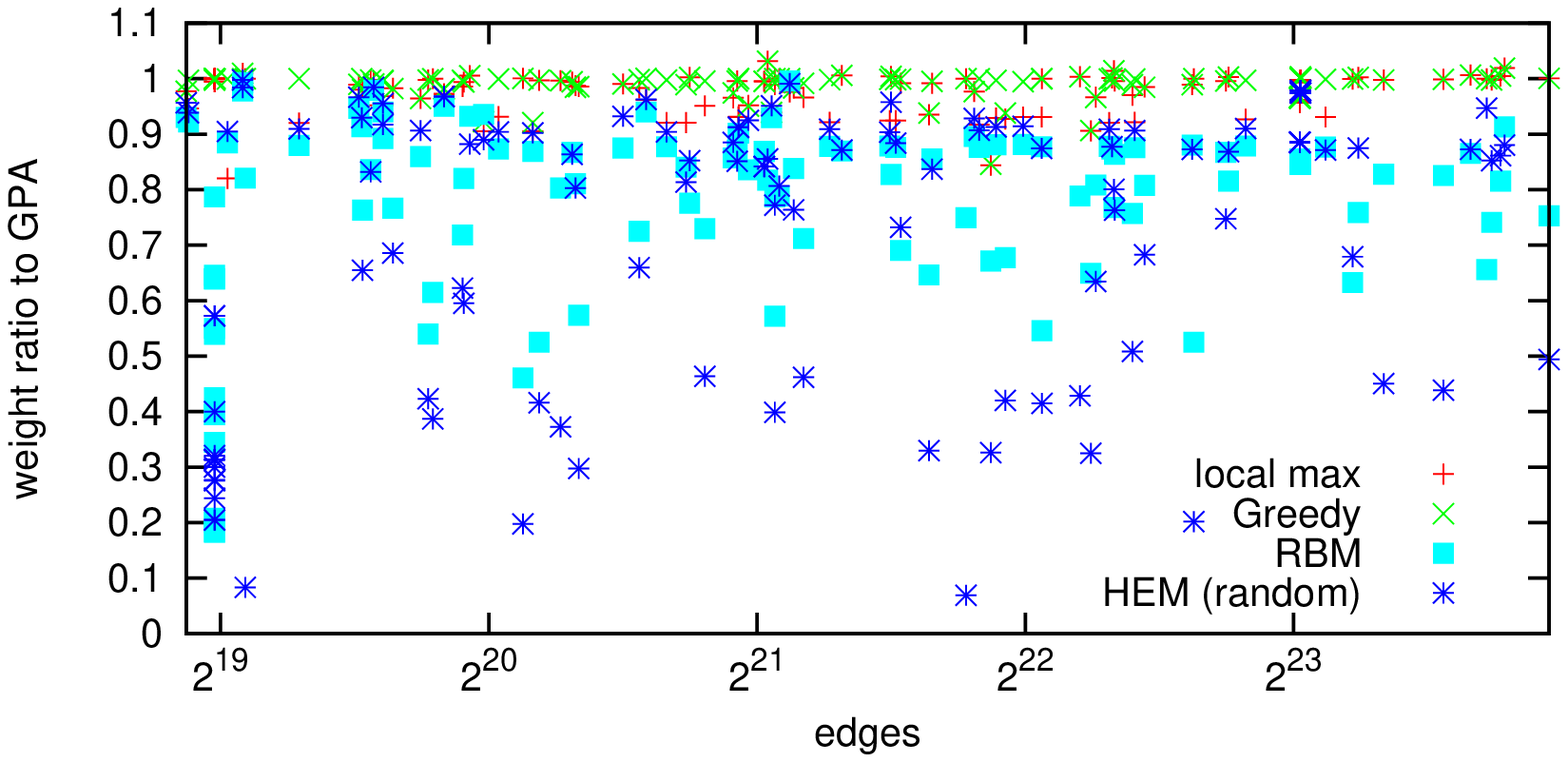}
\includegraphics[scale=0.7]{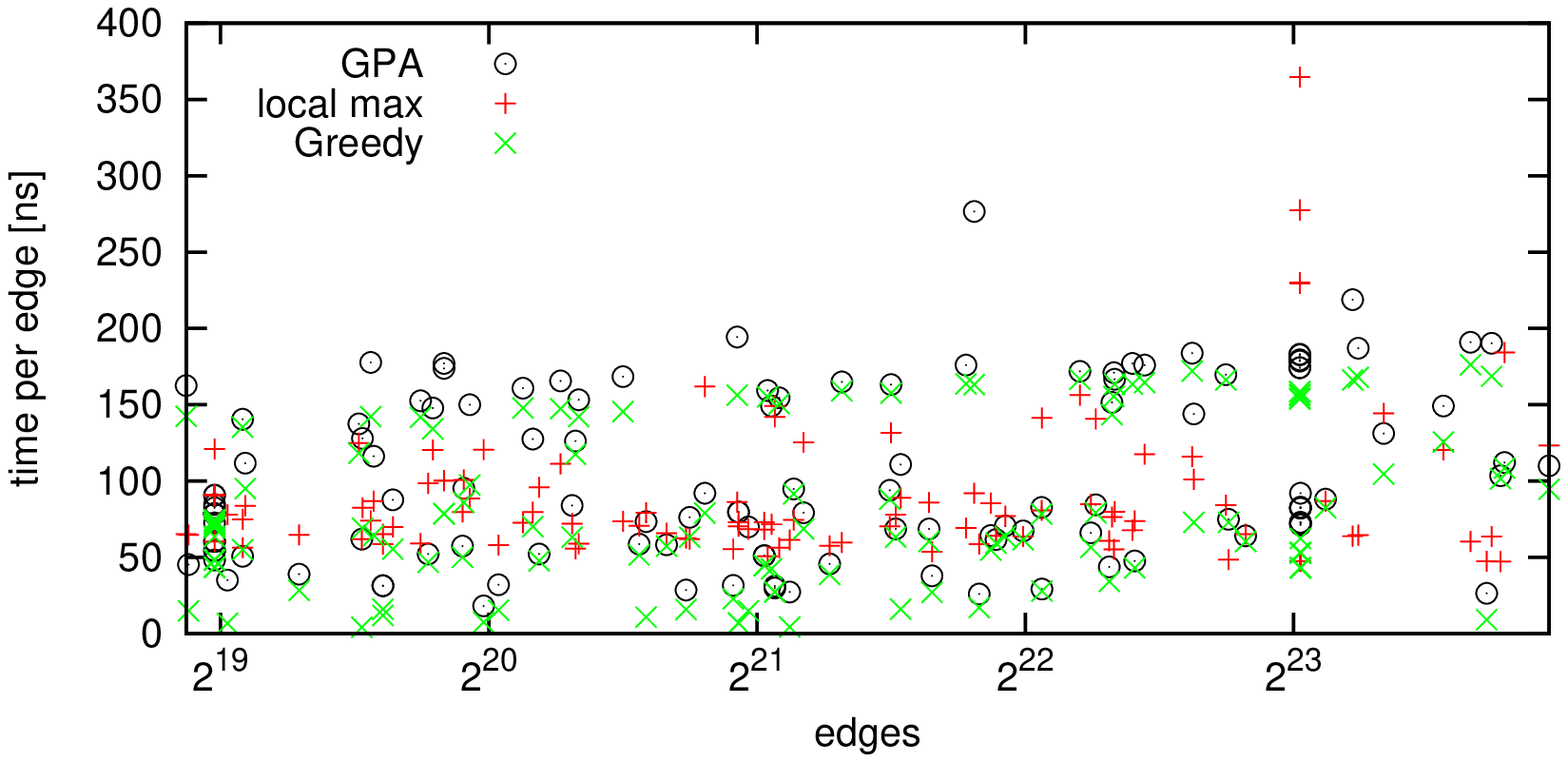}
\caption{\label{fig:floridaquality}Ratio of the weights computed by GPA and other sequential algorithms for sparse matrix instances and running time.}
\end{center}
\end{figure}

We compare solution quality of the algorithms relative to GPA. Via the
experiments in \cite{MauSan07} this also allows some comparison with optimal
solutions which are only a few percent better
there. Figure~\ref{fig:delaunayquality} shows the quality for Delaunay graphs
(where GPA is about 5 \% from optimal \cite{MauSan07}). We see that local max
achieves almost the same quality as greedy which is only about 2 \% worse than
GPA. HEM, possibly the fastest nontrivial sequential algorithm is about 13 \%
away while RBM is 14 \% worse than GPA, i.e., HEM and RBM almost double the gap
to optimality of local max.  Looking at the running times, we see that HEM is the fastest (with a surprisingly
large cost for actually randomizing node orders) followed by local max, greedy,
GPA, and RBM.  From this it looks like HEM, local max, and GPA are the winners
in the sense that none of them is dominated by another algorithm with respect to
both quality and running time. Greedy has similar quality as local max but takes
somewhat longer and is not so easy to parallelize. RBM as a sequential algorithm
is dominated by all other algorithms. Perhaps the most surprising thing is that
RBM is fairly slow. This has to be taken into account when evaluating reported
speedups. We suspect that a more efficient implementation is possible but do not
expect that this changes the overall conclusion.  In Appendix~\ref{app:moredata}
we report similar results for the rgg instances (Figure~\ref{fig:rggquality}) and random graphs (Figures~\ref{fig:random4},~\ref{fig:random16},~\ref{fig:random64}).

Looking at the wide range of instances in the Florida Sparse Matrix collection
leads to similar but more complicated conclusions.
Figure~\ref{fig:floridaquality} shows the solution qualities for greedy, local
max, RBM and HEM relative to GPA. RBM and even more so HEM shows erratic behavior with respect to solution
quality. Greedy and local max are again very close to GPA and even closer to
each other although there is a sizable minority of instances where greedy is
somewhat better than local max. Looking at the corresponding running times one gets a surprisingly diverse picture. HEM which
is again fastest and RBM which is again dominated by local max are not shown.
There are instances where local max is
considerably faster then greedy and vice versa. A possible explanation is that
greedy becomes quite fast when there is only a small number of different edge
weights since then sorting is a quite easy problem.

Experiments on the graph contraction
instances in \cite{Birn12} show local max about 1 \% away from GPA.  For these
instances the average fraction of remaining edges after an iteration is well
below 25 \%. Notable exceptions are the graphs \Id{add20} and \Id{memplus} which
both represent VLSI circuits. Nevertheless, none of the instances considered required more than 10 iterations.

\subsection{Distributed Memory Implementation}\label{s:distributed}
\begin{figure}[b]
\begin{center}
	\subfigure{
	\begin{tikzpicture}[yscale=.8,xscale=.525]
		\begin{scope}[x=1pt,y=1pt]
\definecolor[named]{drawColor}{rgb}{0.00,0.00,0.00}
\definecolor[named]{fillColor}{rgb}{1.00,1.00,1.00}
\fill[color=fillColor,fill opacity=0.00,] (0,0) rectangle (307.15,180.67);
\begin{scope}
\path[clip] ( 59.26, 59.26) rectangle (276.79,166.22);
\definecolor[named]{drawColor}{rgb}{0.00,0.00,0.00}

\draw[color=drawColor,dash pattern=on 1pt off 3pt ,line cap=round,line join=round,fill opacity=0.00,] ( 67.32,134.45) --
	( 96.09,134.45) --
	(124.87,134.45) --
	(153.64,134.45) --
	(182.41,134.45) --
	(211.19,134.45) --
	(239.96,134.45) --
	(268.74,134.45);

\draw[color=drawColor,dash pattern=on 1pt off 3pt ,line cap=round,line join=round,fill opacity=0.00,] ( 67.32, 73.40) --
	( 96.09, 73.40) --
	(124.87, 73.40) --
	(153.64, 73.40) --
	(182.41, 73.40) --
	(211.19, 73.40) --
	(239.96, 73.40) --
	(268.74, 73.40);

\draw[color=drawColor,line width= 0.8pt,line cap=round,line join=round,fill opacity=0.00,] ( 67.32,124.42) --
	( 96.09,119.22) --
	(124.87, 89.16) --
	(153.64,134.54) --
	(182.41,135.46) --
	(211.19,137.27) --
	(239.96,131.95) --
	(268.74, 89.65);

\draw[color=drawColor,line width= 0.8pt,line cap=round,line join=round,fill opacity=0.00,] ( 65.07,122.17) rectangle ( 69.57,126.67);

\draw[color=drawColor,line width= 0.8pt,line cap=round,line join=round,fill opacity=0.00,] ( 93.84,116.97) rectangle ( 98.34,121.47);

\draw[color=drawColor,line width= 0.8pt,line cap=round,line join=round,fill opacity=0.00,] (122.62, 86.91) rectangle (127.12, 91.41);

\draw[color=drawColor,line width= 0.8pt,line cap=round,line join=round,fill opacity=0.00,] (151.39,132.29) rectangle (155.89,136.79);

\draw[color=drawColor,line width= 0.8pt,line cap=round,line join=round,fill opacity=0.00,] (180.16,133.21) rectangle (184.66,137.71);

\draw[color=drawColor,line width= 0.8pt,line cap=round,line join=round,fill opacity=0.00,] (208.94,135.02) rectangle (213.44,139.52);

\draw[color=drawColor,line width= 0.8pt,line cap=round,line join=round,fill opacity=0.00,] (237.71,129.70) rectangle (242.21,134.20);

\draw[color=drawColor,line width= 0.8pt,line cap=round,line join=round,fill opacity=0.00,] (266.49, 87.40) rectangle (270.99, 91.90);
\end{scope}
\begin{scope}
\path[clip] (  0.00,  0.00) rectangle (307.15,180.67);
\definecolor[named]{drawColor}{rgb}{0.00,0.00,0.00}

\draw[color=drawColor,line cap=round,line join=round,fill opacity=0.00,] ( 67.32, 59.26) -- (268.74, 59.26);

\draw[color=drawColor,line cap=round,line join=round,fill opacity=0.00,] ( 67.32, 59.26) -- ( 67.32, 53.26);

\draw[color=drawColor,line cap=round,line join=round,fill opacity=0.00,] ( 96.09, 59.26) -- ( 96.09, 53.26);

\draw[color=drawColor,line cap=round,line join=round,fill opacity=0.00,] (124.87, 59.26) -- (124.87, 53.26);

\draw[color=drawColor,line cap=round,line join=round,fill opacity=0.00,] (153.64, 59.26) -- (153.64, 53.26);

\draw[color=drawColor,line cap=round,line join=round,fill opacity=0.00,] (182.41, 59.26) -- (182.41, 53.26);

\draw[color=drawColor,line cap=round,line join=round,fill opacity=0.00,] (211.19, 59.26) -- (211.19, 53.26);

\draw[color=drawColor,line cap=round,line join=round,fill opacity=0.00,] (239.96, 59.26) -- (239.96, 53.26);

\draw[color=drawColor,line cap=round,line join=round,fill opacity=0.00,] (268.74, 59.26) -- (268.74, 53.26);

\node[color=drawColor,anchor=base,inner sep=0pt, outer sep=0pt, scale=  1.00] at ( 67.32, 35.26) {$2$};

\node[color=drawColor,anchor=base,inner sep=0pt, outer sep=0pt, scale=  1.00] at ( 96.09, 35.26) {$4$};

\node[color=drawColor,anchor=base,inner sep=0pt, outer sep=0pt, scale=  1.00] at (124.87, 35.26) {$8$};

\node[color=drawColor,anchor=base,inner sep=0pt, outer sep=0pt, scale=  1.00] at (153.64, 35.26) {$16$};

\node[color=drawColor,anchor=base,inner sep=0pt, outer sep=0pt, scale=  1.00] at (182.41, 35.26) {$32$};

\node[color=drawColor,anchor=base,inner sep=0pt, outer sep=0pt, scale=  1.00] at (211.19, 35.26) {$64$};

\node[color=drawColor,anchor=base,inner sep=0pt, outer sep=0pt, scale=  1.00] at (239.96, 35.26) {$128$};

\node[color=drawColor,anchor=base,inner sep=0pt, outer sep=0pt, scale=  1.00] at (268.74, 35.26) {$256$};

\draw[color=drawColor,line cap=round,line join=round,fill opacity=0.00,] ( 59.26, 73.40) -- ( 59.26,164.97);

\draw[color=drawColor,line cap=round,line join=round,fill opacity=0.00,] ( 59.26, 73.40) -- ( 53.26, 73.40);

\draw[color=drawColor,line cap=round,line join=round,fill opacity=0.00,] ( 59.26,103.92) -- ( 53.26,103.92);

\draw[color=drawColor,line cap=round,line join=round,fill opacity=0.00,] ( 59.26,134.45) -- ( 53.26,134.45);

\draw[color=drawColor,line cap=round,line join=round,fill opacity=0.00,] ( 59.26,164.97) -- ( 53.26,164.97);

\node[rotate= 90.00,color=drawColor,anchor=base,inner sep=0pt, outer sep=0pt, scale=  1.00] at ( 47.26, 73.40) {$1$};

\node[rotate= 90.00,color=drawColor,anchor=base,inner sep=0pt, outer sep=0pt, scale=  1.00] at ( 47.26,103.92) {$1.5$};

\node[rotate= 90.00,color=drawColor,anchor=base,inner sep=0pt, outer sep=0pt, scale=  1.00] at ( 47.26,134.45) {$2$};

\node[rotate= 90.00,color=drawColor,anchor=base,inner sep=0pt, outer sep=0pt, scale=  1.00] at ( 47.26,164.97) {$2.5$};

\draw[color=drawColor,line cap=round,line join=round,fill opacity=0.00,] ( 59.26, 59.26) --
	(276.79, 59.26) --
	(276.79,166.22) --
	( 59.26,166.22) --
	( 59.26, 59.26);
\end{scope}
\begin{scope}
\path[clip] (  0.00,  0.00) rectangle (307.15,180.67);
\definecolor[named]{drawColor}{rgb}{0.00,0.00,0.00}

\node[color=drawColor,anchor=base,inner sep=0pt, outer sep=0pt, scale=  1.00] at (168.03, 11.26) {Number of processes $p$};

\node[rotate= 90.00,color=drawColor,anchor=base,inner sep=0pt, outer sep=0pt, scale=  1.00] at ( 23.26,112.74) {runtime($\frac{p}{2}$)/runtime($p$)};
\end{scope}
		\end{scope}
	\end{tikzpicture}
	}
	\subfigure{
	\begin{tikzpicture}[yscale=.8,xscale=.525]
		\begin{scope}[x=1pt,y=1pt]
\definecolor[named]{drawColor}{rgb}{0.00,0.00,0.00}
\definecolor[named]{fillColor}{rgb}{1.00,1.00,1.00}
\fill[color=fillColor,fill opacity=0.00,] (0,0) rectangle (307.15,180.67);
\begin{scope}
\path[clip] ( 59.26, 59.26) rectangle (276.79,166.22);
\definecolor[named]{fillColor}{rgb}{0.00,0.00,0.00}
\definecolor[named]{drawColor}{rgb}{0.00,0.00,0.00}

\draw[color=drawColor,dash pattern=on 1pt off 3pt ,line cap=round,line join=round,fill opacity=0.00,] ( 67.32,136.97) --
	( 96.09,136.97) --
	(124.87,136.97) --
	(153.64,136.97) --
	(182.41,136.97) --
	(211.19,136.97) --
	(239.96,136.97) --
	(268.74,136.97);

\draw[color=drawColor,dash pattern=on 1pt off 3pt ,line cap=round,line join=round,fill opacity=0.00,] ( 67.32, 73.76) --
	( 96.09, 73.76) --
	(124.87, 73.76) --
	(153.64, 73.76) --
	(182.41, 73.76) --
	(211.19, 73.76) --
	(239.96, 73.76) --
	(268.74, 73.76);

\draw[color=drawColor,line width= 0.8pt,line cap=round,line join=round,fill opacity=0.00,] ( 67.32,130.68) --
	( 96.09,117.03) --
	(124.87, 91.11) --
	(153.64,136.64) --
	(182.41,136.91) --
	(211.19,136.62) --
	(239.96,135.00) --
	(268.74,127.10);

\draw[color=drawColor,line width= 0.8pt,line cap=round,line join=round,fill opacity=0.00,] ( 65.07,128.43) rectangle ( 69.57,132.93);

\draw[color=drawColor,line width= 0.8pt,line cap=round,line join=round,fill opacity=0.00,] ( 93.84,114.78) rectangle ( 98.34,119.28);

\draw[color=drawColor,line width= 0.8pt,line cap=round,line join=round,fill opacity=0.00,] (122.62, 88.86) rectangle (127.12, 93.36);

\draw[color=drawColor,line width= 0.8pt,line cap=round,line join=round,fill opacity=0.00,] (151.39,134.39) rectangle (155.89,138.89);

\draw[color=drawColor,line width= 0.8pt,line cap=round,line join=round,fill opacity=0.00,] (180.16,134.66) rectangle (184.66,139.16);

\draw[color=drawColor,line width= 0.8pt,line cap=round,line join=round,fill opacity=0.00,] (208.94,134.37) rectangle (213.44,138.87);

\draw[color=drawColor,line width= 0.8pt,line cap=round,line join=round,fill opacity=0.00,] (237.71,132.75) rectangle (242.21,137.25);

\draw[color=drawColor,line width= 0.8pt,line cap=round,line join=round,fill opacity=0.00,] (266.49,124.85) rectangle (270.99,129.35);
\end{scope}
\begin{scope}
\path[clip] (  0.00,  0.00) rectangle (307.15,180.67);
\definecolor[named]{fillColor}{rgb}{0.00,0.00,0.00}
\definecolor[named]{drawColor}{rgb}{0.00,0.00,0.00}

\draw[color=drawColor,line cap=round,line join=round,fill opacity=0.00,] ( 67.32, 59.26) -- (268.74, 59.26);

\draw[color=drawColor,line cap=round,line join=round,fill opacity=0.00,] ( 67.32, 59.26) -- ( 67.32, 53.26);

\draw[color=drawColor,line cap=round,line join=round,fill opacity=0.00,] ( 96.09, 59.26) -- ( 96.09, 53.26);

\draw[color=drawColor,line cap=round,line join=round,fill opacity=0.00,] (124.87, 59.26) -- (124.87, 53.26);

\draw[color=drawColor,line cap=round,line join=round,fill opacity=0.00,] (153.64, 59.26) -- (153.64, 53.26);

\draw[color=drawColor,line cap=round,line join=round,fill opacity=0.00,] (182.41, 59.26) -- (182.41, 53.26);

\draw[color=drawColor,line cap=round,line join=round,fill opacity=0.00,] (211.19, 59.26) -- (211.19, 53.26);

\draw[color=drawColor,line cap=round,line join=round,fill opacity=0.00,] (239.96, 59.26) -- (239.96, 53.26);

\draw[color=drawColor,line cap=round,line join=round,fill opacity=0.00,] (268.74, 59.26) -- (268.74, 53.26);

\node[color=drawColor,anchor=base,inner sep=0pt, outer sep=0pt, scale=  1.00] at ( 67.32, 35.26) {$2$};

\node[color=drawColor,anchor=base,inner sep=0pt, outer sep=0pt, scale=  1.00] at ( 96.09, 35.26) {$4$};

\node[color=drawColor,anchor=base,inner sep=0pt, outer sep=0pt, scale=  1.00] at (124.87, 35.26) {$8$};

\node[color=drawColor,anchor=base,inner sep=0pt, outer sep=0pt, scale=  1.00] at (153.64, 35.26) {$16$};

\node[color=drawColor,anchor=base,inner sep=0pt, outer sep=0pt, scale=  1.00] at (182.41, 35.26) {$32$};

\node[color=drawColor,anchor=base,inner sep=0pt, outer sep=0pt, scale=  1.00] at (211.19, 35.26) {$64$};

\node[color=drawColor,anchor=base,inner sep=0pt, outer sep=0pt, scale=  1.00] at (239.96, 35.26) {$128$};

\node[color=drawColor,anchor=base,inner sep=0pt, outer sep=0pt, scale=  1.00] at (268.74, 35.26) {$256$};

\draw[color=drawColor,line cap=round,line join=round,fill opacity=0.00,] ( 59.26, 73.76) -- ( 59.26,136.97);

\draw[color=drawColor,line cap=round,line join=round,fill opacity=0.00,] ( 59.26, 73.76) -- ( 53.26, 73.76);

\draw[color=drawColor,line cap=round,line join=round,fill opacity=0.00,] ( 59.26,105.37) -- ( 53.26,105.37);

\draw[color=drawColor,line cap=round,line join=round,fill opacity=0.00,] ( 59.26,136.97) -- ( 53.26,136.97);

\node[rotate= 90.00,color=drawColor,anchor=base,inner sep=0pt, outer sep=0pt, scale=  1.00] at ( 47.26, 73.76) {$1$};

\node[rotate= 90.00,color=drawColor,anchor=base,inner sep=0pt, outer sep=0pt, scale=  1.00] at ( 47.26,105.37) {$1.5$};

\node[rotate= 90.00,color=drawColor,anchor=base,inner sep=0pt, outer sep=0pt, scale=  1.00] at ( 47.26,136.97) {$2$};

\draw[color=drawColor,line cap=round,line join=round,fill opacity=0.00,] ( 59.26, 59.26) --
	(276.79, 59.26) --
	(276.79,166.22) --
	( 59.26,166.22) --
	( 59.26, 59.26);
\end{scope}
\begin{scope}
\path[clip] (  0.00,  0.00) rectangle (307.15,180.67);
\definecolor[named]{fillColor}{rgb}{0.00,0.00,0.00}
\definecolor[named]{drawColor}{rgb}{0.00,0.00,0.00}

\node[color=drawColor,anchor=base,inner sep=0pt, outer sep=0pt, scale=  1.00] at (168.03, 11.26) {Number of processes $p$};

\end{scope}
		\end{scope}
	\end{tikzpicture}
	}

	\caption{Scaling results of the parallel local max algorithm on random geometric graphs with random
	edge weights. Left: rgg23 ($\approx $63 million edges). Right: rgg24 ($\approx$ 132 million edges).}
\label{fig:rggdistributed}
\end{center}
\end{figure}
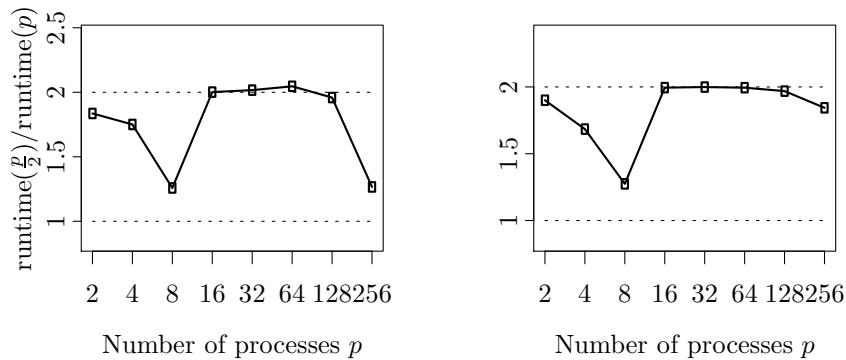

Our distributed memory parallelization (using MPI) on $p$ processing elements
(PEs or MPI processes) assigns nodes to PEs and stores all edges incident to a
node locally. This can be done in a load balanced way if no node has degree
exceeding $m/p$.  The second pass of the basic algorithm from
Section~\ref{s:parallel} has to exchange information on candidate edges that
cross a PE boundary. In the worst case, this can involve all edges handled by a
PE, i.e., we can expect better performance if we manage to keep most edges
locally. In our experiments, one PE owns nodes whose numbers are a consecutive
range of the input numbers. Thus, depending on how much locality the input
numbering contains we have a highly local or a highly non-local situation.  We
have not considered more sophisticated ways of node assignment so far since our
motivating application is graph partitioning/clustering where almost by
definition we initially do \emph{not} know which nodes form clusters -- this is
the intended \emph{output}. Since Lemma~\ref{lem:remove} also applies to the
subgraph relevant for a particular PE, we can expect that the graph shrinks
fairly uniformly over the entire network.

We performed experiments on two different clusters at the KIT computing center
both using compute-nodes with two quad-core processors each. Refer to \cite{Birn12} for details.
We ran experiments with up 128 compute-nodes corresponding to
1024 cores with one MPI process per core.

Figure~\ref{fig:rggdistributed} illustrates how our distributed local max
implementation scales for the random geometric graphs \Id{rgg23} and \Id{rgg24}
(using random edge weights) which have fairly good locality. We plot the decrease in
running time for successive doubling of $p$, i.e., a value of two stands
for perfect relative speedup for this step and a value below one means that
parallelization no longer helps. We see values slightly below two for the steps
$1\rightarrow 2$ and $2\rightarrow 4$ which is typical behavior of multicore
algorithms when cores compete for resources like memory bandwidth. For $p=8$ we
start to use two compute-nodes (with 4 active cores each) and consequently we see
the largest dip in efficiency. Beyond that, we have almost perfect scaling until
the problem instance becomes too small. We have similar behavior for other
graphs with good locality. For graphs with poor locality, efficiency is not very
good. However the ratios stay above one for a very long time, i.e., it pays to
use parallelism when it is available anyway. This is the situation we have when
partitioning large graphs for use on massively parallel machines. Considering
that the matching step in graph partitioning is often the least work intensive
one in multi-level graph partitioning algorithms we conclude that local max
might be a way to remove a sequential bottleneck from massively parallel graph
partitioning. Refer to \cite{Birn12} for additional data.

\subsection{GPU Implementation}\label{s:gpu}
\begin{figure}
\begin{center}
\includegraphics[scale=0.7]{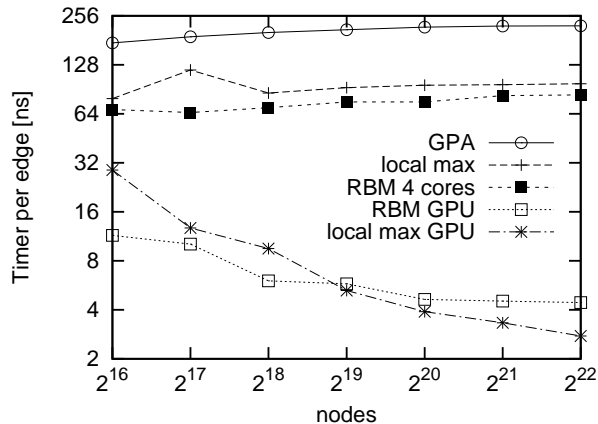}
\caption{\label{fig:delaunaygpu}Running time of sequential and GPU algorithms for Delaunay instances.}
\end{center}
\end{figure}

Our GPU algorithm is a fairly direct implementation of the CRCW algorithm.
We reduce the algorithm to the basic primitives such as segmented prefix sum, prefix sum and random gather/scatter from/to GPU memory.
As a basis for our implementation we use back40computing library by Merrill \cite{back40computing}.

Figure~\ref{fig:delaunaygpu} compares the running time of our implementation
with GPA, sequential local max, the RBM algorithm parallelized for 4 cores, and
its GPU parallelization from \cite{FagBis12}. While the CPU implementation has
troubles recovering from its sequential inefficiency and is only slightly faster
than even sequential local max, the GPU implementation is impressively fast in
particular for small graphs. For large graphs, the GPU implementation of local
max is faster.  Since local max has better solution quality, we consider this a
good result.  Our GPU code is up to 35 times faster than sequential local max.
We may also be able to learn from the implementation techniques of RBM GPU for
small inputs in future work.

For random geometric graphs and random graphs, we get similar behavior (see
Figure~\ref{fig:rgggpu}~and~\ref{fig:randomgpu} in Appendix~\ref{app:moredata}).
The results for rgg are slightly worse for GPU local max -- speedup is up to 24 over sequential local max
and a speed advantage over GPU RBM only for the very largest inputs.
As for random graphs, the denser the graph is the larger is our speedup over the sequential and GPU RBM implementations.
Thus, for $\alpha=64$ our implementation is faster than GPU RBM for $n=2^{15}$ already.
While for $n=2^{18}$ it is $65\%$ faster than GPU RBM and 30 times faster than the sequential local max.

\section{Conclusions}

The local max algorithm is a good choice for parallel or external computation of
maximal and approximate maximum weight matchings. On the theoretical side it is
provably efficient for computing maximal matchings and guarantees a
1/2-approximation. On the practical side it yields better quality at faster
speed than several competitors including the greedy algorithm and RBM. Somewhat
surprisingly it is even attractive as a sequential algorithm, outperforming HEM
with respect to solution quality and other algorithms with respect to speed.

Many interesting question remain. Can we omit re-randomization of edge
weights when computing maximal matchings? Is there a linear work parallel
algorithm with polylogarithmic execution time that computes 1/2-approximations
(or any other constant factor approximation). Can we even do 2/3-approximations
with linear work in parallel \cite{DH03c,PS04b}?

\bibliographystyle{plain}
\bibliography{diss}
\clearpage
\begin{appendix}
\section{More Experimental results}\label{app:moredata}
\begin{figure}[h]
\begin{center}
\includegraphics[scale=0.8]{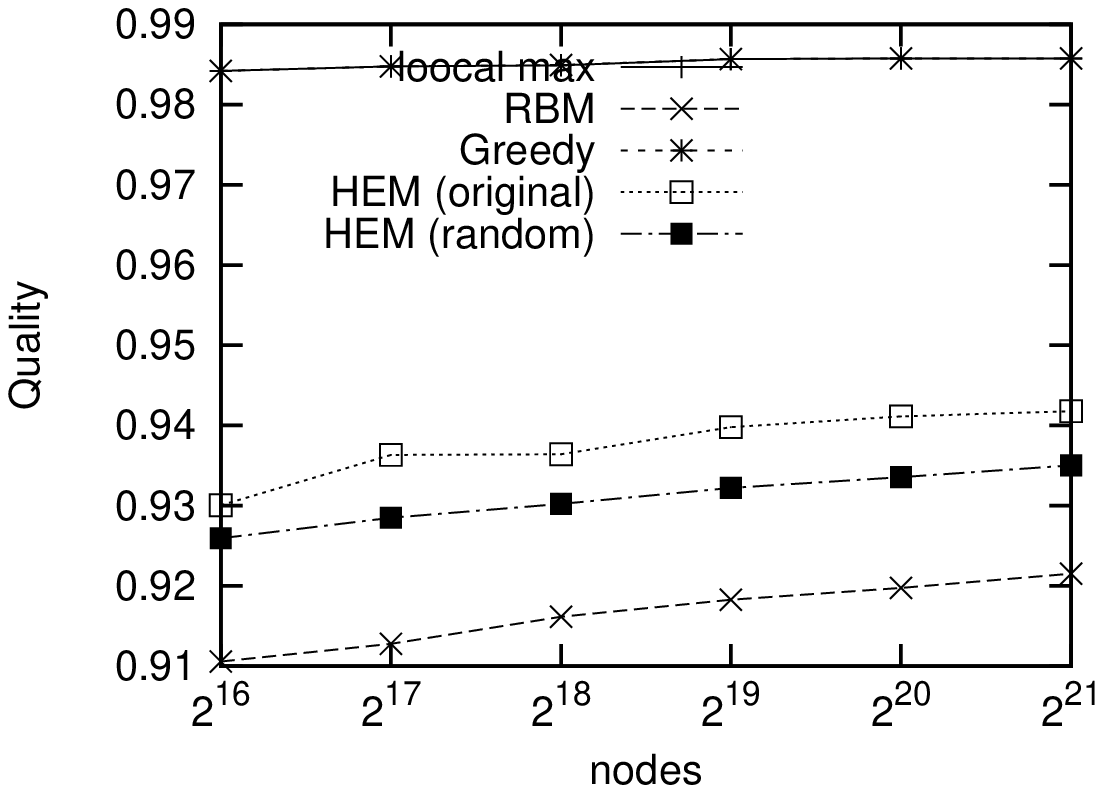}
\includegraphics[scale=0.8]{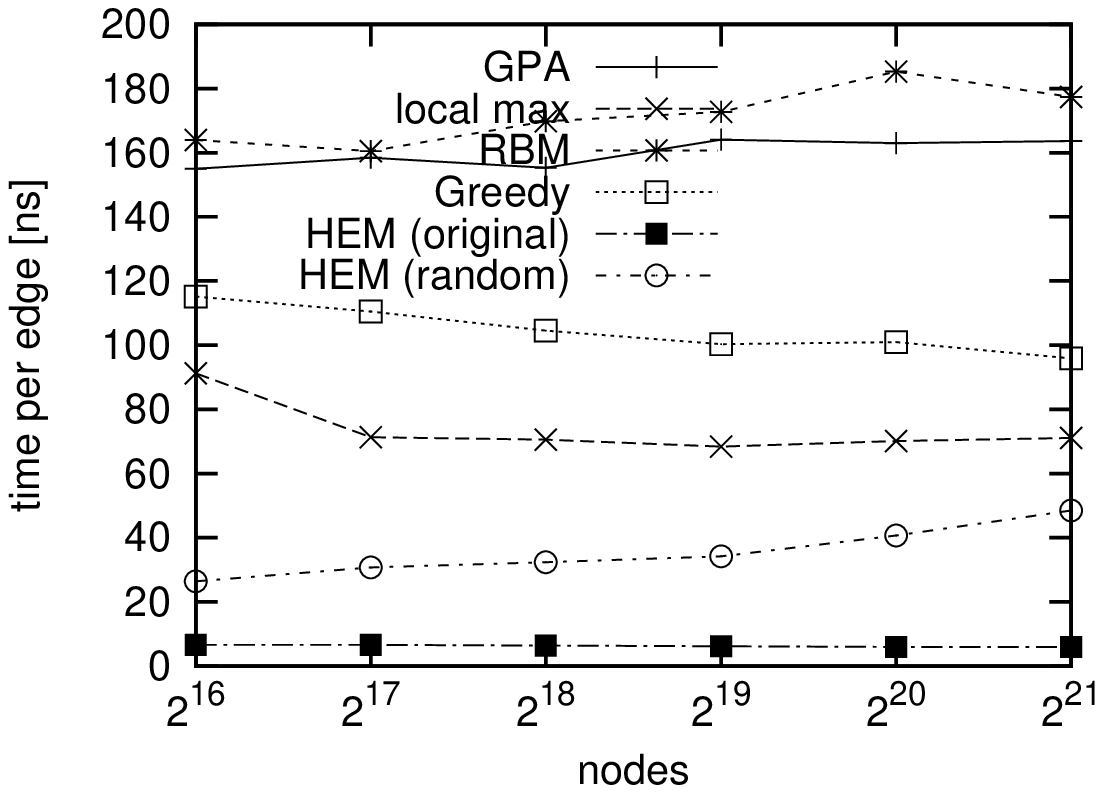}
\caption{\label{fig:rggquality}Ratio of the weights computed by GPA and other algorithms for random geometric graphs and running time.}
\end{center}
\end{figure}
\begin{figure}
\begin{center}
\includegraphics[scale=0.5]{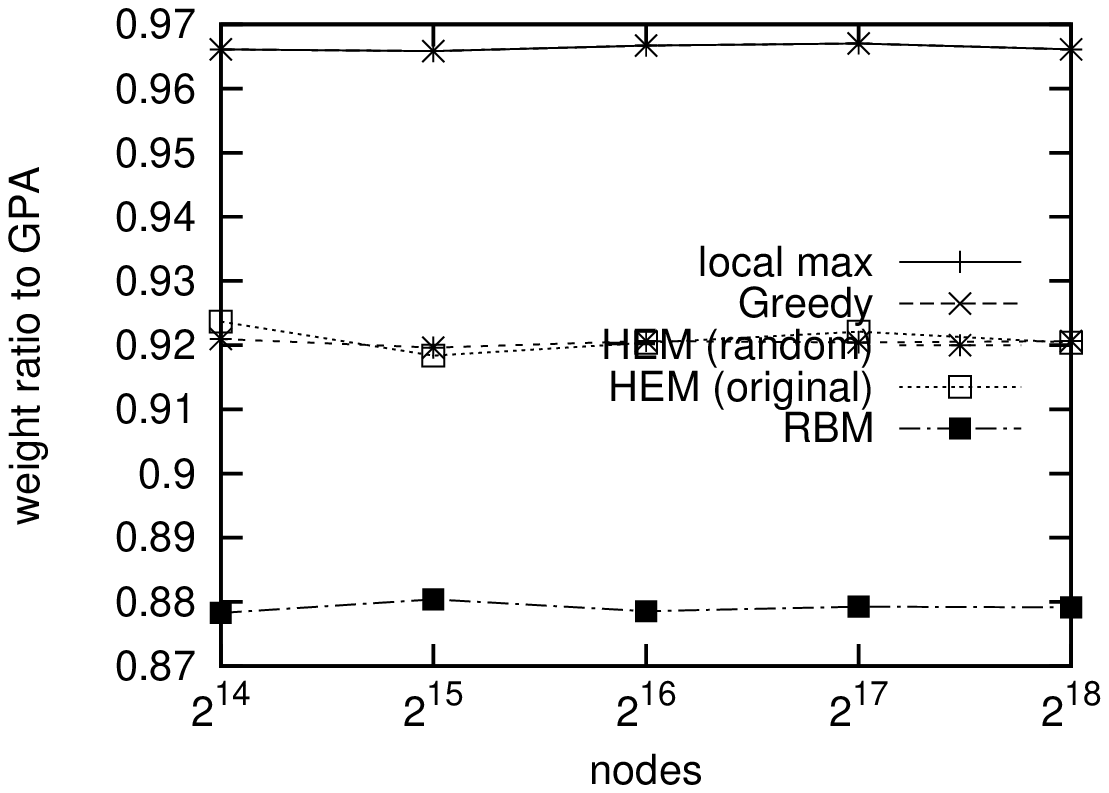}~\includegraphics[scale=0.5]{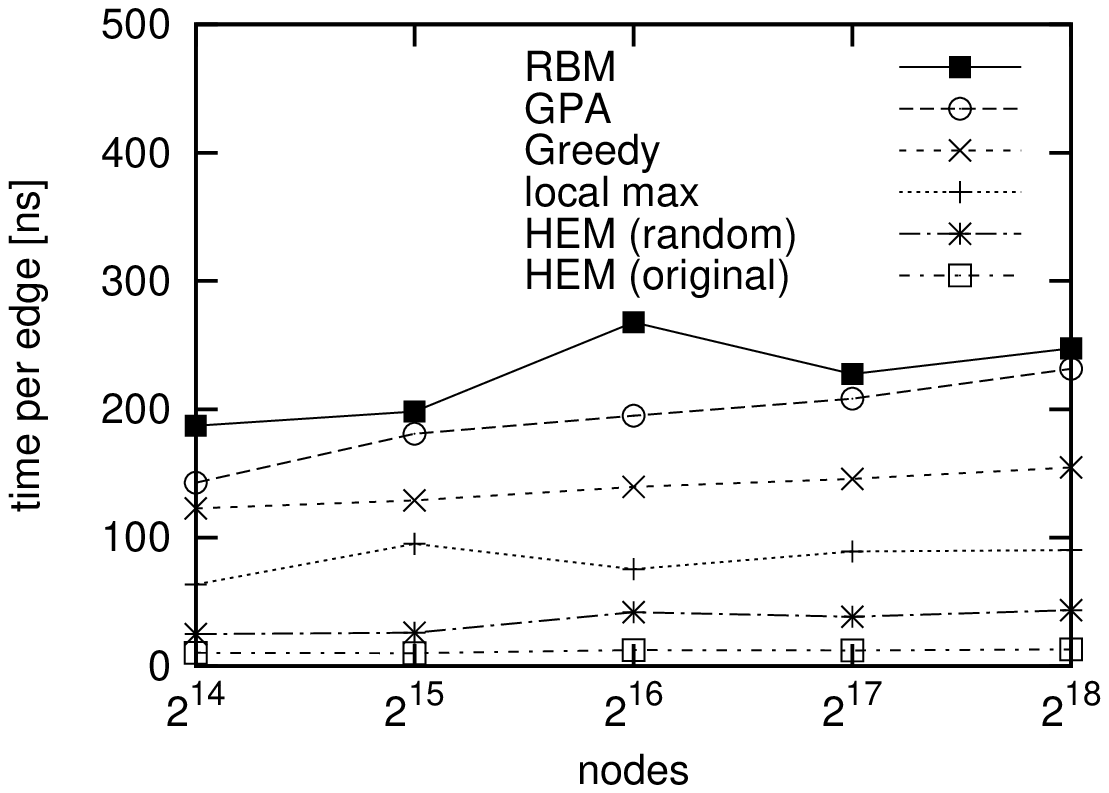}
\caption{\label{fig:random4}Ratio of the weights computed by GPA and other sequential algorithms (left) and their timing (right) for random graphs with $\alpha=4$.}
\end{center}
\end{figure}
\begin{figure}
\begin{center}
\includegraphics[scale=0.5]{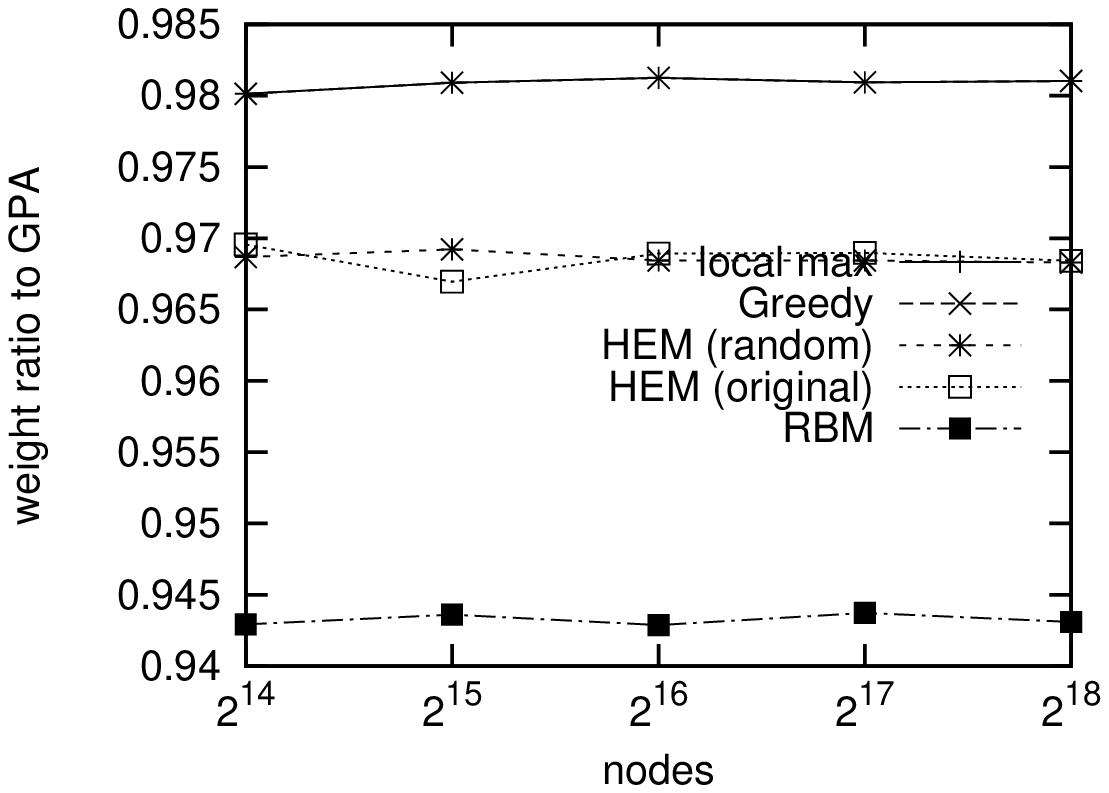}~\includegraphics[scale=0.5]{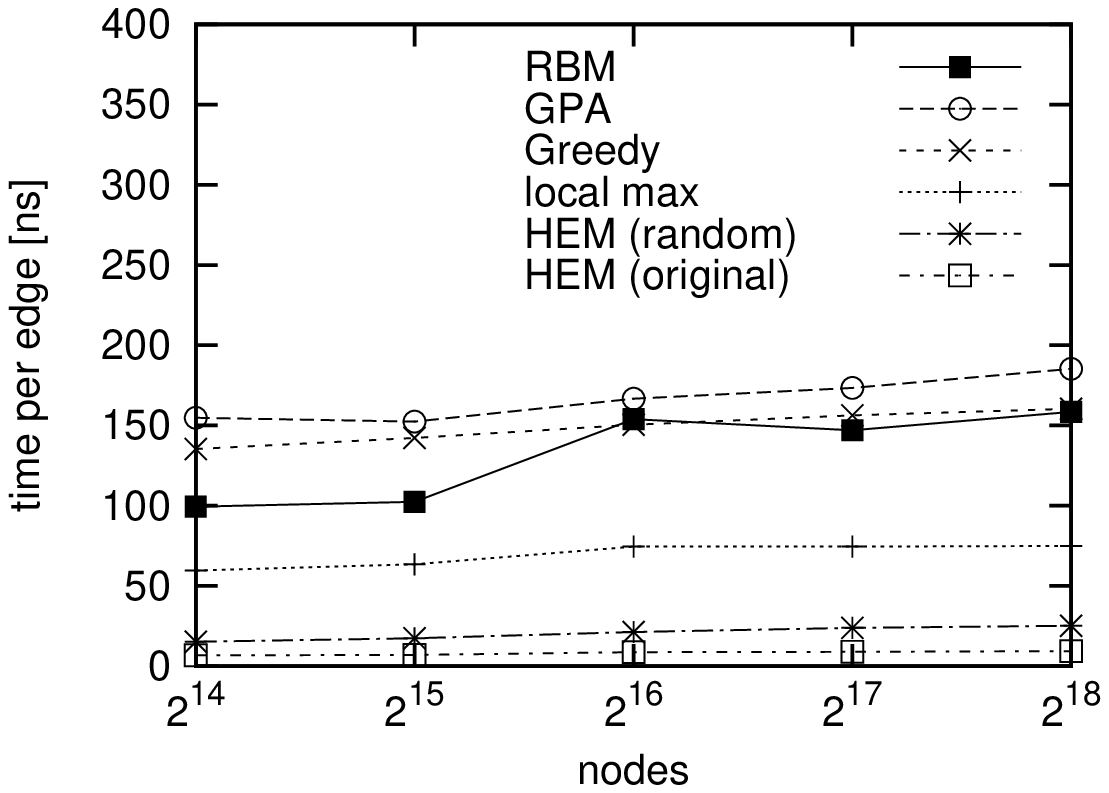}
\caption{\label{fig:random16}Ratio of the weights computed by GPA and other sequential algorithms (left) and their timing (right) for random graphs with $\alpha=16$.}
\end{center}
\end{figure}
\begin{figure}
\begin{center}
\includegraphics[scale=0.5]{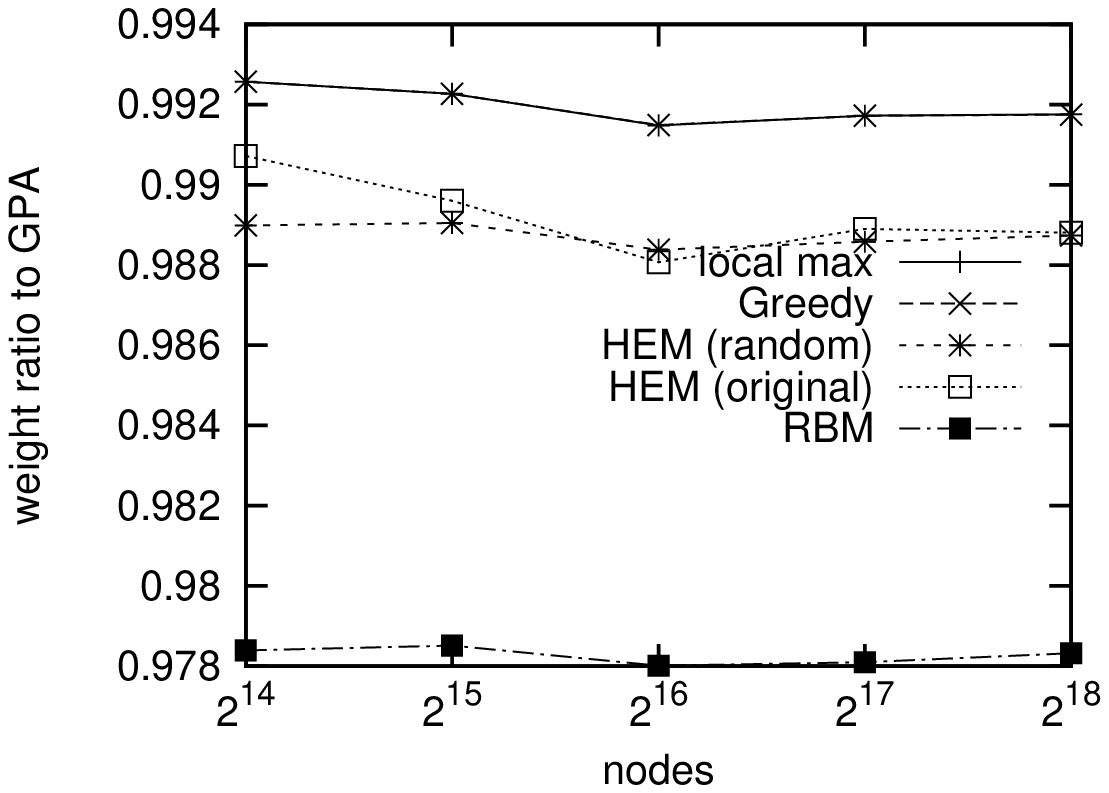}~\includegraphics[scale=0.5]{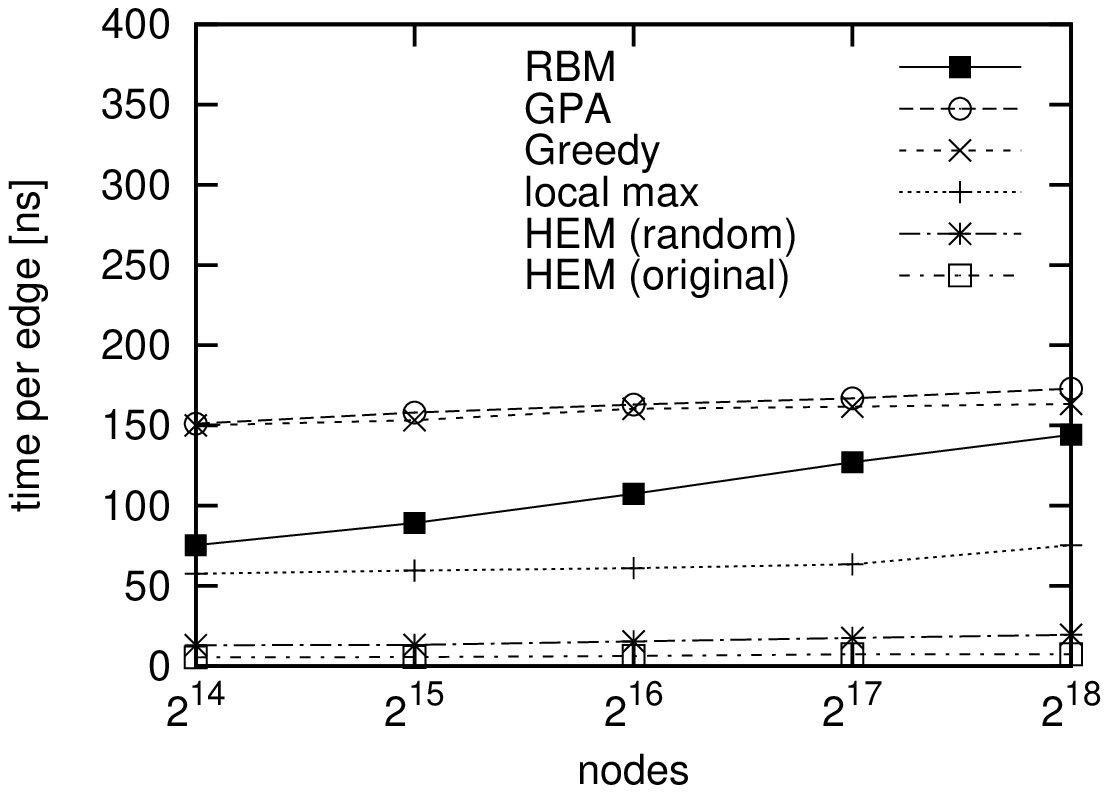}
\caption{\label{fig:random64}Ratio of the weights computed by GPA and other sequential algorithms (left) and their timing (right) for random graphs with $\alpha=64$.}
\end{center}
\end{figure}
\begin{figure}
\begin{center}
\includegraphics[scale=0.8]{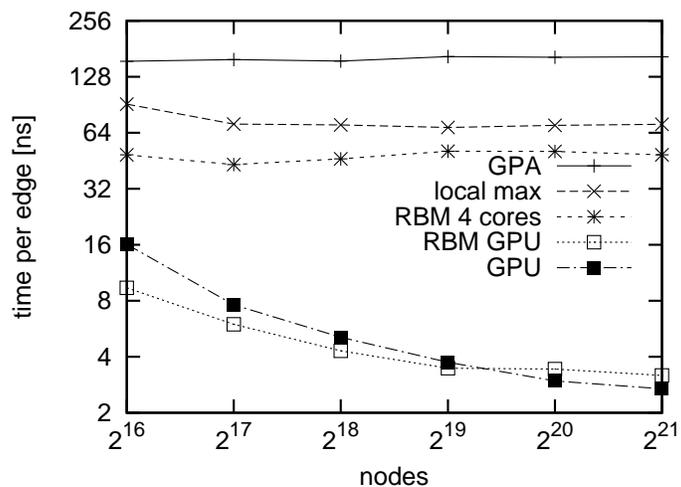}
\caption{\label{fig:rgggpu}Running time of sequential and GPU algorithms for random geometric graphs instances.}
\end{center}
\end{figure}
\begin{figure}
\begin{center}
\includegraphics[scale=0.5]{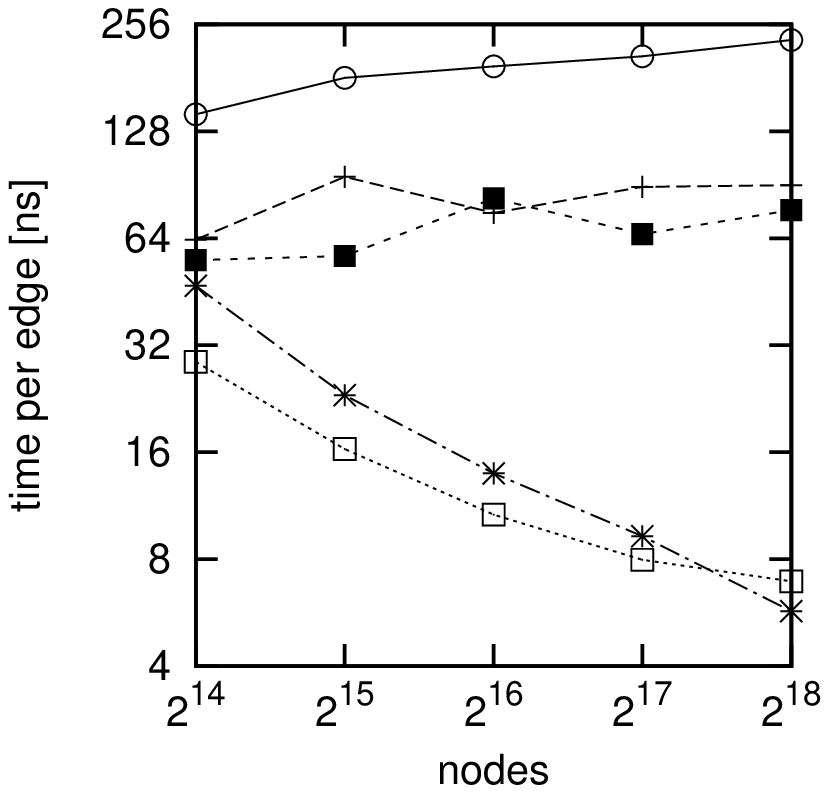}~\includegraphics[scale=0.5]{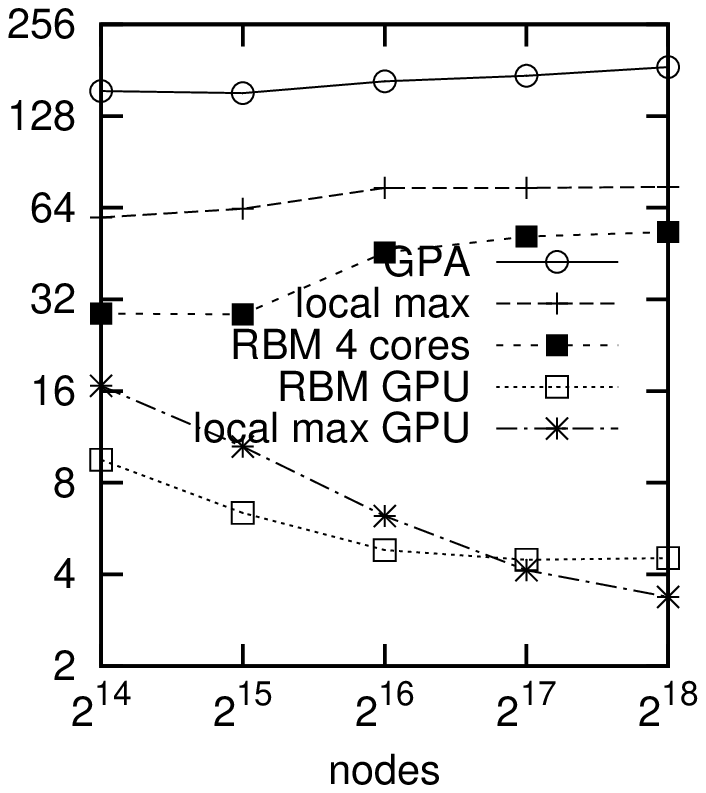}~\includegraphics[scale=0.5]{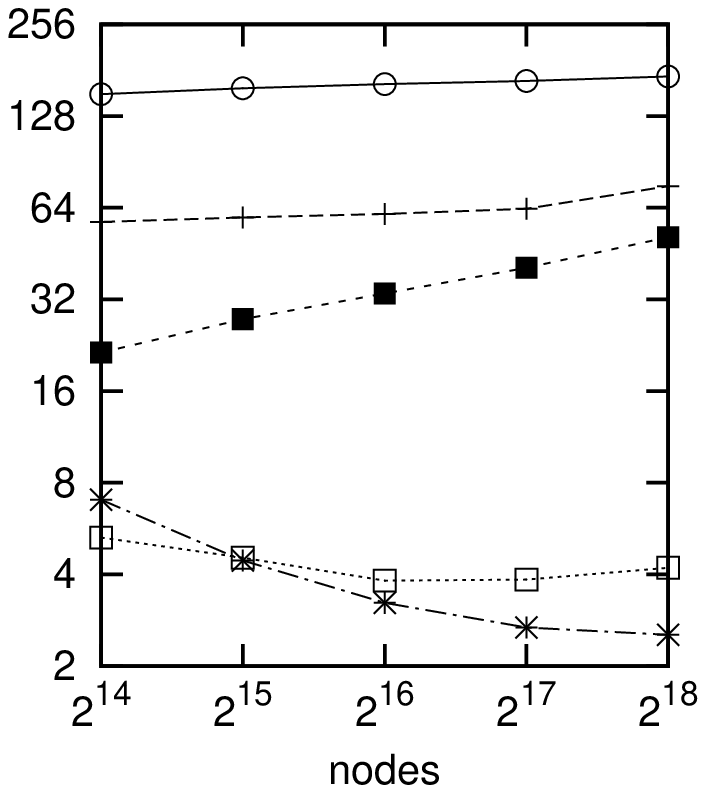}
\caption{\label{fig:randomgpu}Running time of sequential and GPU algorithms for random graphs with $\alpha=4,16,64$ (from left to right).}
\end{center}
\end{figure}
\section{List of Instances}\label{app:instances}
\begin{tabular}{|l|r|r|}
  \hline
file&n&m\\
  \hline
2cubes\_sphere.mtx.graph&101492&772886\\
af\_0\_k101.mtx.graph&503625&8523525\\
af\_1\_k101.mtx.graph&503625&8523525\\
af\_2\_k101.mtx.graph&503625&8523525\\
af\_3\_k101.mtx.graph&503625&8523525\\
af\_4\_k101.mtx.graph&503625&8523525\\
af\_5\_k101.mtx.graph&503625&8523525\\
af\_shell1.mtx.graph&504855&8542010\\
af\_shell2.mtx.graph&504855&8542010\\
af\_shell3.mtx.graph&504855&8542010\\
af\_shell4.mtx.graph&504855&8542010\\
af\_shell5.mtx.graph&504855&8542010\\
af\_shell6.mtx.graph&504855&8542010\\
af\_shell7.mtx.graph&504855&8542010\\
af\_shell8.mtx.graph&504855&8542010\\
af\_shell9.mtx.graph&504855&8542010\\
apache2.mtx.graph&715176&2051347\\
BenElechi1.mtx.graph&245874&6452311\\
bmw3\_2.mtx.graph&227362&5530634\\
bmw7st\_1.mtx.graph&141347&3599160\\
bmwcra\_1.mtx.graph&148770&5247616\\
boneS01.mtx.graph&127224&3293964\\
boyd1.mtx.graph&93279&558985\\
c-73.mtx.graph&169422&554926\\
c-73b.mtx.graph&169422&554926\\
c-big.mtx.graph&345241&997885\\
cant.mtx.graph&62451&1972466\\
case39.mtx.graph&40216&516021\\
case39\_A\_01.mtx.graph&40216&516021\\
case39\_A\_02.mtx.graph&40216&516026\\
case39\_A\_03.mtx.graph&40216&516026\\
case39\_A\_04.mtx.graph&40216&516026\\
case39\_A\_05.mtx.graph&40216&516026\\
case39\_A\_06.mtx.graph&40216&516026\\
case39\_A\_07.mtx.graph&40216&516026\\
case39\_A\_08.mtx.graph&40216&516026\\
case39\_A\_09.mtx.graph&40216&516026\\
case39\_A\_10.mtx.graph&40216&516026\\
case39\_A\_11.mtx.graph&40216&516026\\
case39\_A\_12.mtx.graph&40216&516026\\
case39\_A\_13.mtx.graph&40216&516026\\
cfd1.mtx.graph&70656&878854\\
cfd2.mtx.graph&123440&1482229\\
CO.mtx.graph&221119&3722469\\
consph.mtx.graph&83334&2963573\\
cop20k\_A.mtx.graph&99843&1262244\\
crankseg\_1.mtx.graph&52804&5280703\\
crankseg\_2.mtx.graph&63838&7042510\\
ct20stif.mtx.graph&52329&1323067\\
  \hline
\end{tabular}~
\begin{tabular}{|l|r|r|}
  \hline
file&n&m\\
  \hline
darcy003.mtx.graph&389874&933557\\
dawson5.mtx.graph&51537&479620\\
denormal.mtx.graph&89400&533412\\
dielFilterV2clx.mtx.graph&607232&12351020\\
dielFilterV3clx.mtx.graph&420408&16232900\\
Dubcova2.mtx.graph&65025&482600\\
Dubcova3.mtx.graph&146689&1744980\\
d\_pretok.mtx.graph&182730&756256\\
ecology1.mtx.graph&1000000&1998000\\
ecology2.mtx.graph&999999&1997996\\
engine.mtx.graph&143571&2281251\\
F1.mtx.graph&343791&13246661\\
F2.mtx.graph&71505&2611390\\
Fault\_639.mtx.graph&638802&13987881\\
filter3D.mtx.graph&106437&1300371\\
G3\_circuit.mtx.graph&1585478&3037674\\
Ga10As10H30.mtx.graph&113081&3001276\\
Ga19As19H42.mtx.graph&133123&4375858\\
Ga3As3H12.mtx.graph&61349&2954799\\
Ga41As41H72.mtx.graph&268096&9110190\\
GaAsH6.mtx.graph&61349&1660230\\
gas\_sensor.mtx.graph&66917&818224\\
Ge87H76.mtx.graph&112985&3889605\\
Ge99H100.mtx.graph&112985&4169205\\
gsm\_106857.mtx.graph&589446&10584739\\
H2O.mtx.graph&67024&1074856\\
helm2d03.mtx.graph&392257&1174839\\
hood.mtx.graph&220542&5273947\\
IG5-17.mtx.graph&30162&1034600\\
invextr1\_new.mtx.graph&30412&906915\\
kkt\_power.mtx.graph&2063494&6482320\\
Lin.mtx.graph&256000&755200\\
mario002.mtx.graph&389874&933557\\
mixtank\_new.mtx.graph&29957&982542\\
mouse\_gene.mtx.graph&45101&14461095\\
msdoor.mtx.graph&415863&9912536\\
m\_t1.mtx.graph&97578&4827996\\
nasasrb.mtx.graph&54870&1311227\\
nd12k.mtx.graph&36000&7092473\\
nd24k.mtx.graph&72000&14321817\\
nlpkkt80.mtx.graph&1062400&13821136\\
offshore.mtx.graph&259789&1991442\\
oilpan.mtx.graph&73752&1761718\\
parabolic\_fem.mtx.graph&525825&1574400\\
pdb1HYS.mtx.graph&36417&2154174\\
pwtk.mtx.graph&217918&5708253\\
qa8fk.mtx.graph&66127&797226\\
qa8fm.mtx.graph&66127&797226\\
  \hline
\end{tabular}

\begin{tabular}{|l|r|r|}
  \hline
  file&n&m\\
  \hline
s3dkq4m2.mtx.graph&90449&2365221\\
s3dkt3m2.mtx.graph&90449&1831506\\
shipsec1.mtx.graph&140874&3836265\\
shipsec5.mtx.graph&179860&4966618\\
shipsec8.mtx.graph&114919&3269240\\
ship\_001.mtx.graph&34920&2304655\\
ship\_003.mtx.graph&121728&3982153\\
Si34H36.mtx.graph&97569&2529405\\
Si41Ge41H72.mtx.graph&185639&7412813\\
Si87H76.mtx.graph&240369&5210631\\
SiO.mtx.graph&33401&642127\\
SiO2.mtx.graph&155331&5564086\\
sparsine.mtx.graph&50000&749494\\
StocF-1465.mtx.graph&1465137&9770126\\
t3dh.mtx.graph&79171&2136467\\
t3dh\_a.mtx.graph&79171&2136467\\
thermal2.mtx.graph&1228045&3676134\\
thread.mtx.graph&29736&2220156\\
tmt\_sym.mtx.graph&726713&2177124\\
TSOPF\_FS\_b162\_c3.mtx.graph&30798&896688\\
TSOPF\_FS\_b162\_c4.mtx.graph&40798&1193898\\
TSOPF\_FS\_b300.mtx.graph&29214&2196173\\
TSOPF\_FS\_b300\_c1.mtx.graph&29214&2196173\\
TSOPF\_FS\_b300\_c2.mtx.graph&56814&4376395\\
TSOPF\_FS\_b39\_c19.mtx.graph&76216&979241\\
TSOPF\_FS\_b39\_c30.mtx.graph&120216&1545521\\
turon\_m.mtx.graph&189924&778531\\
vanbody.mtx.graph&47072&1144913\\
x104.mtx.graph&108384&5029620\\
  \hline
\end{tabular}

\end{appendix}
\end{document}